\documentclass[a4paper,UKenglish]{article}
\usepackage[utf8]{inputenc}
\usepackage{hyperref}
\usepackage{amsmath,amssymb,amsthm}
\usepackage[ruled]{algorithm2e}
\usepackage{cleveref}
\usepackage{thm-restate}
\usepackage[margin=1in]{geometry}
\usepackage{newtxtext,newtxmath}
\usepackage{booktabs}
\usepackage{authblk}
\usepackage{dirtytalk}
\usepackage[dvipsnames]{xcolor}

\usepackage{listings}
\usepackage{enumitem}
\setlist[1]{labelindent=\parindent}
\setlist[enumerate]{label=(\arabic*),noitemsep}
\setlist[itemize]{noitemsep}

\newtheorem{theorem}{Theorem}
\newtheorem{lemma}[theorem]{Lemma}
\newtheorem{corollary}[theorem]{Corollary}

\newtheorem{definition}[theorem]{Definition}
\newtheorem*{definition*}{Definition}

\newtheorem{observation}[theorem]{Observation}
\theoremstyle{definition}

\theoremstyle{remark}

\newtheorem*{note*}{Note}

\newtheorem*{remark*}{Remark}
\newtheorem{open}{Open Problem}

\title{Revisiting Diameter in Directed Graphs}
\author[1,2]{Ben Bals}
\author[1]{Joakim Blikstad}
\author[1]{Daniel Dadush}
\author[1,2]{Yasamin Nazari}
\author[3]{Jonas Schmidt}
\affil[1]{CWI, Amsterdam, The Netherlands}
\affil[2]{Vrije Universiteit, Amsterdam, The Netherlands}
\affil[3]{Bocconi University, Milan, Italy}

\usepackage{thmtools}
\usepackage{thm-restate}
\usepackage{placeins}
\usepackage[normalem]{ulem}
\usepackage{cryptocode}
\usepackage{bm}
\usepackage{tabularx}
\usepackage{makecell}

\title{Revisiting Diameter in Directed Graphs} 
\crefname{theorem}{Thm.}{Thms.}
\Crefname{theorem}{Theorem}{Theorems}
\crefname{corollary}{Cor.}{Cors.}
\Crefname{corollary}{Corollary}{Corollaries}

\ifdefined\N\else\DeclareMathOperator{\N}{\mathbb{N}}\fi
\ifdefined\Z\else\DeclareMathOperator{\Z}{\mathbb{Z}}\fi
\ifdefined\R\else\DeclareMathOperator{\R}{\mathbb{R}}\fi

\newcommand{\BigO}{\mathcal{O}}
\newcommand{\OTilde}{\tilde{\BigO}}
\newcommand{\Oish}{\OTilde}
\newcommand{\setsize}[1]{\left|#1\right|}

\newcommand{\tw}{\mathrm{tw}}

\DeclareMathOperator{\tdOp}{ReachDiam}
\newcommand{\td}[1]{\tdOp(#1)}

\DeclareMathOperator{\tc}{TC}
\DeclareMathOperator{\MinDiam}{MinDiam}
\newcommand{\smol}{\varepsilon}

\newif\ifpaper

\newcommand{\tail}[1]{\mathrm{tail}(e)}

\DeclareMathOperator{\mcc}{MCC}

\newcommand{\ie}[1]{(i.e.,~#1)}
\newcommand{\eg}[1]{(e.g.,~#1)}

\newif\ifhidetodos
\newcommand{\cbox}[2][yellow]{%
  \fcolorbox{#1}{white}{\parbox{\dimexpr\linewidth-2\fboxsep}{\strut #2\strut}}%
}

\newcommand{\customtodo}[3]{\textcolor{#2}{\(\blacktriangledown\)}\marginnote{\raggedright \textcolor{#2}{\textbf{#1:} #3}}}
\newcommand{\customtododisplay}[3]{\noindent\cbox[#2]{\textcolor{#2}{\textbf{#1:} #3}}}
\newcommand{\customtodoinline}[3]{\textcolor{#2}{\textcolor{#2}{\textbf{#1:} #3}}}

\ifhidetodos
    \renewcommand{\customtodo}[3]{}
    \renewcommand{\customtododisplay}[3]{}
    \renewcommand{\customtodoinline}[3]{}
\fi

\usepackage{caption}
\usepackage{subcaption}
\usepackage{mathtools}
\usepackage[ruled]{algorithm2e}

\usepackage{thmtools}
\usepackage{thm-restate}

\usepackage{booktabs}
\newtheorem{hypothesis}{Hypothesis}
\newtheorem{proofsketch}{Proof Sketch}

\usepackage{comment}

\newif\ifshort
\newif\iflong

\longtrue
\begin{document}

\maketitle

\begin{abstract}
The reachability diameter ($\tdOp$) of a directed graph is the maximum distance over all pairs $u,v$ where $v$ is reachable from $u$.
This notion is present in the definition of shortcut sets, and the name was recently coined in that context by Haeupler, Jiang, and Saranurak~[SOSA 2026].
While this is a very natural notion of diameter in directed graphs, and especially DAGs, it is so far not computationally explored. Other definitions of diameter in directed graphs are either trivial (infinite) in graphs that are not strongly connected \eg{the classical definition} or are non-trivial only in highly restrictive graph classes \eg{Min-Diameter}.

We initiate the problem of computing the (approximate) reachability diameter from a fine-grained complexity point of view.
Under certain fine-grained assumptions, we prove that there is no algorithm in time $\BigO(n^{\omega - \smol}$) that gives \emph{any} approximation of $\tdOp$ in \emph{weighted graphs}. Similarly, there is no algorithm with better than $2$-approximation for \emph{unweighted graphs} in this time. To supplement this, we provide algorithmic upper bounds that lead to additive approximation of $\tdOp$ for unweighted graphs.
Hence, we establish a strong separation between the weighted and unweighted cases, which makes this type of diameter different in nature than other known notions.

Considering the hardness in general weighted graphs, we also study special graph classes and get small constant approximations for DAGs with bounded width or graphs with bounded treewidth.
Interestingly, our techniques also lead to exact hopsets with hopbound $2$ for bounded treewidth graphs.
This and some of our upper bounds for general graphs show technical connections between approximating $\tdOp$ and computing shortcut sets and hopsets.
\end{abstract}

\section{Introduction}

Computing the diameter of a graph is a fundamental problem in fine-grained complexity that has been widely studied both from an upper bound and a lower bound perspective. While the problem is very well-understood in undirected graphs, there is still much that we do not know for directed graphs. In particular, it is not even clear what the \textit{right definition} of diameter for directed graphs should be, depending on the application. For example, existing definitions of the diameter become trivial in a DAG. In this work, we focus on the notion of the \textit{reachability diameter} of a directed graph recently defined by \cite{SOSA}. Although this is a very natural definition, it has not yet been computationally explored. Our goal in this work is to characterize the fine-grained complexity of computing or approximating this notion.

We first formally define it:
\begin{definition*}[Reachability Diameter]
    \label{def:td}
    Let $G$ be a directed graph.  Let $\td{G} \coloneqq \max_{(u,v) \in \tc(G)} d(u,v)$, where $\tc(G)$ is the transitive closure of $G$. In other words, we consider the maximum shortest path distance over all pairs $u, v \in V(G)$ where $v$ is reachable from $u$. 
\end{definition*}
This definition is closely related to the notion of diameter used for shortcut sets. In particular, \cite{SOSA} used low diameter decompositions based on this notion of diameter for constructing shortcut sets. Informally, a shortcut set is a set of edges $H$ added to a graph $G$ such that the reachability diameter of $G \cup H$ is bounded by a given parameter~$\beta$, called the hopbound. While shortcut sets and their distance preserving variants, hopsets, have found many applications and received significant attention recently, surprisingly \textit{computing} this relevant notion of diameter has, to the best of our knowledge, not yet been studied. 

In addition to fine-grained lower bounds, we will show additive approximations for this problem that lead to good enough estimates when this diameter is large. This is the relevant regime for shortcut sets in general directed graphs, in which our algorithms can be used to \textit{verify} the diameter reduction procedure after adding shortcut edges. However, we see that getting constant approximate solutions in the other regimes turns out to be hard. 

First, let us discuss how this definition compares to other directed diameter definitions and how they compare in the difficulty of computation.

\paragraph{Classical Diameter.} Classically, the diameter of a directed graph is the largest shortest path distance over all pairs of vertices.  This is infinite when the graph is not strongly connected. 
This diameter can be estimated in a straight-forward way. Consider the following folklore approach: Pick an arbitrary vertex $u$. Compute an inward and an outward Dijkstra. Return the largest distance from or to $u$. It is easy to show that this is a $2$-approximation for the classical diameter in weighted directed graphs. This also correctly returns $\infty$ when the graph is not strongly-connected. 

The state-of-the-art for this diameter \cite{roditty_fast_2013} improves this to a $3/2$-approximation by a more sophisticated argument that involves sampling a set of vertices and exploring a limited number of vertices from this set. In fact, for this notion of diameter, this approach provides similar guaranties on directed and undirected graphs.
 
With such approaches, if for a pair of vertices $u,v$ either of $d(u,v)$ or $d(v,u)$ is infinite, the algorithm returns infinity, which is the correct answer for this notion of diameter. Hence for graphs that are not strongly connected, this value is easy to compute but not very insightful as we do not receive any information about the vertex pairs that have a finite distance.
In such graphs, we would like to have a notion of distance for each pair in the relevant direction, namely the direction where those vertices are reachable. This is exactly what the reachability diameter captures. 

Note that if the classical diameter is finite \ie{in a strongly connected graph}, then it is equivalent to the reachability diameter. 
Therefore reachability diameter can be seen as a generalization of the classical notion.
In particular, classical diameter reduces to reachability diameter simply by first checking for strong connectivity. 

At a high level, this is why we expect the reachability diameter to be more difficult to compute.
We make this formal in
\iflong
\Cref{sec:weighted}
\fi
\ifshort
the full version
\fi
where, among other things, we prove that there exists no constant factor approximation of reachability diameter in weighted graphs in near-linear time unless well-known fine-grained conjectures fail.

\paragraph{Round-trip Diameter.} The roundtrip diameter is another diameter notion that attempts to capture the maximum distance in both directions: \cite{abboud_approximation_2016} describes it as ``the distance between two vertices $u$ and $v$ corresponds to going from $u$ to $v$ and back''. Formally, the round-trip diameter is defined as the maximum over all pairs $u,v$ of $d(u,v) + d(v,u)$. This notion also has the same limitation that it is only finite in strongly connected graphs.
\paragraph{Min-Diameter.}
The min-diameter ($\MinDiam$) was introduced to offer a meaningful notion of diameter in a directed graph that is not strongly connected and has received considerable attention \cite{abboud_approximation_2016,DBLP:conf/icalp/DalirrooyfardW019,DBLP:conf/focs/ChechikZ22,DBLP:conf/esa/BergerKW23}. The min-diameter of a directed graph is defined as the maximum over all vertices $u,v$ of $\min(d(u,v), d(v,u))$. 

The goal with this definition is somewhat similar to ours, as this value will be finite as long as there is reachability in one direction for each pair. Hence, this value is still meaningful for graphs that are not strongly connected.
However, this definition is still quite restrictive. For example, a DAG has a finite min-diameter if and only if it represents a total order. Equivalently, such graphs will have the following structure: a single path plus extra edges along the path.
For general directed graphs, including those that are not strongly connected, the DAG of strongly connected components must have such a structure for the $\MinDiam$ to be finite.

Therefore, unless a directed graph has this very specific structure, its $\MinDiam$ will again be infinite and not informative. 
Conceptually, this can also explain why $\MinDiam$ turns out to be easier to compute than $\tdOp$.
On DAGs, $\MinDiam$ can be 2-approximated in near-linear time \cite{abboud_approximation_2016}, which we show is impossible for $\tdOp$ under standard fine-grained conjectures in
\iflong
\Cref{sec:weighted}.
\fi
\ifshort
the full version.
\fi

Note that reachability diameter can also be seen as a generalization of the min-diameter: In the more interesting finite case, $\MinDiam$ is exactly equivalent to $\tdOp$.
The infinite case is easy to detect in linear time by computing strongly connected components and then a topological sort.

Recently, \cite{DBLP:conf/esa/BergerKW23} proposed a 3/2-approximation for $\MinDiam$ of sparse unweighted DAGs in subquadratic time.
Also on weighted general directed graphs, with a more sophisticated approach, $\MinDiam$ can be approximated to a constant factor in near-linear time due to a recent result by Chechik and Zhang \cite{DBLP:conf/focs/ChechikZ22}.

\paragraph{Weighted vs Unweighted Reachability Diameter.}
Another interesting fact about this new notion of diameter that we establish is a strong separation between approximating the weighted vs unweighted case that does not hold for the other definitions. In particular, we show that in weighted graphs assuming certain known fine-grained conjectures, we cannot get any multiplicative approximation in time $\BigO(n^{\omega-\smol})$, whereas we can compute approximations for the unweighted case. This also hints at why reachability diameter is harder to compute than the other notions, as most known diameter estimation techniques work both for the unweighted and the weighted case (possibly with extra technical steps and analysis). 

\paragraph{Difficulty of Applying Known Techniques for Computing the Reachability Diameter.}
For both the classical definition and for $\MinDiam$, the known algorithms involve a step of sampling vertices, and computing the in- and out-reachability from each sampled vertex. This will give an estimate that for the classic diameter leads to a 2-approximation directly, and with some additional steps also leads to a constant-factor approximation for $\MinDiam$ or better approximations of the classical diameter.

For reachability diameter, this scheme has a fundamental limitation: the estimates obtained from the furthest distance of a set of sampled vertices may not give any useful information on $\tdOp$.
First consider the folklore approach for classical diameter that runs an inward and outward Dijkstra from a single vertex $u$. The analysis then resorts to triangle inequality to approximate any distance $d(a,b)$ with $d(a,u) + d(u,b)$. For $\tdOp$ however, if one of the latter distances is infinite, we gain no information about the true distance between $a$ and $b$, not even if it is finite.

Intuitively, a large part of the graph may show a similar behavior: Consider any directed graph with $n$ vertices. Say, we add another $n$ source vertices, with a directed edge to every original vertex. This is enough to turn both classical diameter and min-diameter infinite, while leaving the reachability diameter completely unchanged.
Furthermore, computing shortest paths from the newly added vertices gives us no information about the original graph; the same holds for balls of fixed size or distance, as commonly used in diameter approximation algorithms.
Thus, any constant-factor approximation scheme that relies on the idea of iteratively sampling a vertex and then computing distances from this vertex would need to contain a mechanism to ignore these ``distraction'' vertices. To the best of our knowledge, the existing diameter toolkit cannot handle and distinguish this case.

The more involved algorithms for classical diameter \eg{\cite{roditty_fast_2013}} and $\MinDiam$ \eg{\cite{abboud_approximation_2016} and \cite{DBLP:conf/focs/ChechikZ22}} will still resist adaptation to $\tdOp$ for similar reasons.
This is because, at their core, they still require computing shortest paths from and to a small number of carefully chosen vertices (potentially in a recursive sub-instance or in limited size balls) to be informative about the diameter.

We will see that with techniques based on sampling vertices or hitting sets, we \textit{can} obtain polynomial additive approximations for reachability diameter in unweighted graphs in $o(mn)$ time, where there is a trade-off between the computation time and the approximation. However for the reasons described, it seems quite challenging to get rid of these larger additive factors for general graphs. This implies that we get better approximations only when the reachability diameter is known to be large, but getting good estimates when the diameter is small remains a challenge. We are, however, able to avoid these bad instances and get good approximations for special graph classes by using chain-cover based or tree decomposition based approaches. 

\paragraph{Connections with Shortcut sets and Hopsets.}
As discussed earlier, we noted that the reachability diameter is the appropriate notion of diameter when we consider objects such as shortcut sets and hopsets. In fact, if after computing a shortcut set (or hopset), we would like to verify that we have correctly reduced the diameter to the desired parameter, we would need to estimate the reachability diameter.

Interestingly, while we cannot show any blackbox reductions between computing shortcut sets and reachability diameter, we do establish some technical connections. We see that using the same ideas as in shortcut sets of \cite{DBLP:conf/focs/LiuJS19} gives us algorithms where the approximation factor for reachability diameter correspond to the hopbound achieved in the same time. 

Building on this connection, we give a new exact hopset construction and a simplified algorithm for shortcut sets for graphs with bounded treewidth.
\iffalse
\paragraph{Related Definitions.}
If we look a bit further, the definitions of the radius of a graph (generally the minimum over all vertices of the maximum of the distance to all other vertices) are as numerous as the definitions of the diameter.
Variations again include the source radius, max-radius, min-radius, and roundtrip radius.
Both the roundtrip diameter and the known radius definitions use a similar algorithmic toolkit as the diameter variations we compared to in more detail above and adopting insights from the studies of these is not directly possible for the very same reasons. 
\fi
\section{Our Results}
Our goal is to characterize the fine-grained complexity landscape of (approximate) reachability diameter computation both from a lower bound and upper bound point of view.

For some settings, these results are tight (such as for approximating $\tdOp$ in weighted graphs), for others gaps remain (such as for unweighted graphs).
While the main upper bound techniques do not translate to our setting for constant-approximation of $\tdOp$ due to described technical challenges, we do show several upper bounds either with larger approximation for general graphs or with small approximation for special graph classes. 

We outline the main results here; see \Cref{tab:algos} for a summary of our upper bounds and \Cref{tab:lb} for our lower bounds.
A main distinction for both upper and lower bounds will be if they apply to weighted or unweighted directed graphs.
When we discuss optimality, we always do so up to subpolynomial factors.

\ifshort
Due to space constraints, all proofs are deferred to the full version, which is appended to this document.
\fi

\iflong
\begin{table}[tp]
    \centering
    \begin{tabular}{clclll} \toprule
        \textbf{Weighted?} & \textbf{Approx.} & \textbf{Det.?} & \textbf{Running time} & \textbf{Ref.} & \textbf{Technique} \\ \midrule
        \checkmark & exact & \checkmark & \(\OTilde(nm)\) & & $n$ $\cdot$ BFS \\
        \checkmark & exact & \checkmark & \(\OTilde(m \cdot \mathrm{tw} + m^{1 + o(1)})\) &  \cref{thm:treewidth} & Separator D\&C like \cite{abboud_approximation_2016} \\ 
        \checkmark & \(1+\smol\) & \checkmark & \(\OTilde(n^\omega / \smol)\) &  & Round + APSP-approx\\
        \checkmark & 3 & \checkmark & \(\OTilde(m \cdot |\mcc|)\) & \cref{cor:mcc_approx} & Path-based approach \\
        \midrule
        x & exact & \checkmark & \(\OTilde(n^\omega)\) &  & Exact transitive closure\\
        x & \((1,k)\) & x & \(\OTilde(nm/k)\) & \cref{thm:random-sampling} & Random hitting set\\
        x & \((3, k)\) & \checkmark & \(\OTilde(nm/k)\) & \cref{thm:sampling_derandomized} & $\ell$-cover + path-based \\
        x & \(n^{1/2 + o(1)}\) & x & \(\OTilde(m)\) & \cref{thm:ljs} & Sampling pivots\\
        \bottomrule
    \end{tabular}
    \caption{Overview of approximation algorithms for reachability diameter on DAGs. Det.\,abbreviates deterministic. The approximation is given as a multiplicative factor except for a tuple $(a,b)$ which indicates a $a$-multiplicative and $b$-additive approximation.}
    \label{tab:algos}
\end{table}
\fi
\ifshort
\begin{table}[tp]
    \centering
    \begin{tabular}{clclll} \toprule
        \textbf{Weighted?} & \textbf{Approx.} & \textbf{Det.?} & \textbf{Running time} & \textbf{Ref.} & \textbf{Technique} \\ \midrule
        \checkmark & exact & \checkmark & \(\OTilde(nm)\) & & $n$ $\cdot$ BFS \\
        \checkmark & exact & \checkmark & \(\OTilde(m \cdot \mathrm{tw} + m^{1 + o(1)})\) &  $\star$ & Separator D\&C like \cite{abboud_approximation_2016} \\ 
        \checkmark & \(1+\smol\) & \checkmark & \(\OTilde(n^\omega / \smol)\) &  & Round + APSP-approx\\
        \checkmark & 3 & \checkmark & \(\OTilde(m \cdot |\mcc|)\) & \cref{cor:mcc_approx} & Path-based approach \\
        \midrule
        x & exact & \checkmark & \(\OTilde(n^\omega)\) &  & Exact transitive closure\\
        x & \((1,k)\) & x & \(\OTilde(nm/k)\) & \cref{thm:random-sampling} & Random hitting set\\
        x & \((3, k)\) & \checkmark & \(\OTilde(nm/k)\) & \cref{thm:sampling_derandomized} & $\ell$-cover + path-based \\
        x & \(n^{1/2 + o(1)}\) & x & \(\OTilde(m)\) & \cref{thm:ljs} & Sampling pivots\\
        \bottomrule
    \end{tabular}
    \caption{Overview of approximation algorithms for reachability diameter on DAGs. Det.\,abbreviates deterministic. The approximation is given as a multiplicative factor except for a tuple $(a,b)$ which indicates a $a$-multiplicative and $b$-additive approximation. References marked with a $\star$ are deferred to the full version.}
    \label{tab:algos}
\end{table}
\fi

\subsection{Weighted Graphs}
For weighted graphs, the worst-cast complexity landscape is relatively simple.
Let us start with the main hardness result (see 
\iflong
\Cref{sec:hypotheses}
\fi
\ifshort
the full version
\fi
for the definitions of the hardness assumptions).

\begin{restatable}{theorem}{thmLbSV}
    \label{thm:lb:sv}
    Unless the simplicial vertex conjecture fails, there  is no $poly(n)$-approximation for \(\tdOp\) on weighted directed graphs (even on DAGs) running in time \(\BigO(n^{\omega-\smol})\).
\end{restatable}

The result holds even if the approximation guarantee is a (at most singly exponential) function of the input graph size.
That is, we also rule out $\log n$ or even factor-$n$ approximations.

This means that the optimal combinatorial algorithm is running Dijkstra's algorithm from all vertices (which will calculate $\tdOp$ exactly).
The standard algebraic $(1+\smol)$ All-Pairs-Shortest-Paths approximation~\cite{Zwick2002} can be used to approximate $\tdOp$.
As we only care about the maximum finite distance, we can get rid of the dependency on the largest edge weight using a rounding technique, getting a $\Oish(n^\omega / \smol)$ running time bound. This was previously observed for standard diameter in~\cite{Bringmann2019}.
We show the same hardness bound based on a different hardness assumption, namely, the high-dimensional orthogonal vectors problem. Hence we put the hardness result on a more secure footing in case one of these conjectures turn out to be false. 

\begin{restatable}{theorem}{thmLbHdov}
    \label{thm:lb:hdov}
    Suppose $\omega > 2$ and consider any $\epsilon >0$. Unless high-dimensional OV fails, there is no algorithm with $\text{poly}(n)$-approximation for \(\tdOp\) on weighted graphs running in time \(\BigO(n^{\omega-\smol})\).
\end{restatable}

Hence, under these fine-grained assumptions, the straightforward approximate APSP upper bound matches the lower bound on worst-case instances. It is therefore natural to see if the problem remains hard in certain graph classes.
In particular, we take a parameterized view at the fine-grained complexity of the problem.
This leads us to near-linear constant-factor approximation algorithms given one of these parameters is bounded.
For these graph classes, we thus side-step the general impossibility of constant-factor approximations in less than matrix-multiplication time.

First, for DAGs, we show that if the \emph{width} (the size of the minimum number of chains required to cover all vertices) is bounded, we can compute a $3$-approximation in almost-linear time. 

\begin{restatable}{corollary}{thmMccApprox}\label{cor:mcc_approx}
    There is an algorithm that, given a DAG $G$, computes a 3-approximation for the \(\tdOp\) in time \(\OTilde(\setsize{\mcc} m + m^{1+o(1)})\), where \(\setsize{\mcc}\) is the width of $G$.
\end{restatable}

The width is a standard parameter for studying problems related to reachability in DAGs.
For example, they are also widely considered for studying shortcut sets \cite{DBLP:conf/soda/KoganP23,DBLP:conf/esa/KoganP24}.
The algorithm underlying our result is inspired by the chain-cover-based approaches widely used in recent advancements in that area.
Crucially, for a chain $C$, we can approximate the diameter among all demand pairs that have a shortest path passing through $C$ in $\Oish(m)$ time.

As a second parameter, we look at the treewidth of the underlying undirected graph of the input directed graph.
This parameter was originally considered for diameter problems by Abboud, Vassilevska Williams, and Wang \cite{abboud_approximation_2016} in the context of fixed parameter subquadratic algorithms.
We share their motivation of exploring structural parameters with the intent of overcoming barriers from fine-grained complexity.
In fact, we can lightly adapt one of their algorithms for the classic diameter definition to compute the $\tdOp$ exactly in almost-linear time for graphs of bounded treewidth.
We first give a much simpler version of this algorithm that already yields a 2-approximation of $\tdOp$.

\paragraph{Exact hopsets.} We then adapt this algorithm to compute shortcut sets and \emph{exact} hopsets that have constant hopbound and linear size for graphs of bounded treewidth.
At its core, this result relies on the fact that any graph of bounded treewidth will recursively have small balanced separators.
This yields a recursive technique where the number of shortcut edges relates to the size of these balanced separators.
Formally, we show the following result.

\begin{restatable}{theorem}{thmTwHopset}
  \label{thm:tw-hopset}
    There is an algorithm that given a directed graph $G$ of treewidth $\mathrm{tw}$, computes an exact hopset of hopbound $2$ of size $\OTilde(n \cdot \mathrm{tw})$ in time $\OTilde(m \cdot \mathrm{tw} + m^{1 + o(1)})$.
\end{restatable}

Compare this to the general bounds: a lower bound due to Bodwin and Hoppenworth shows that there are graphs where we cannot construct a $\BigO(n)$-size shortcut set with stretch better than $n^{1/4}$ \cite{DBLP:conf/focs/BodwinH23}.
For exact hopsets, they even show that there are graphs where there cannot be a $\BigO(n)$-size set with hopbound better than $\sqrt{n}$.

Due to the folklore connection between TC-spanners and shortcut sets, a result like \Cref{thm:tw-hopset} was already implied by a more complicated algorithm for $H$-minor-free graphs by \cite{DBLP:journals/siamcomp/BhattacharyyaGJRW12}. Hence we get a simpler algorithm for shortcut sets. For exact hopsets, this is the first such result.

Recently, Chalermsook, Jiang, Mukhopadhyay, and Nanongkai   \cite{DBLP:journals/corr/abs-2502-08032} raised the question of finding efficient algorithms computing better shortcut sets and TC-spanners in restricted graph classes, and explicitly suggested low-treewidth graphs as such a class to explore. 
We observe that the result in [BGJ+] implicitly gives such improved bounds for bounded tree-width graphs, which is also implied by our simplified approach. Moreover, we are addressing this question for hopsets.
Both for shortcut sets and exact hopsets, these provide small constant stretch in the most-commonly studied near-linear size regime.

\begin{table}[tp]
    \centering
    \begin{tabular}{clllll} \toprule
         \textbf{Weighted?} & \textbf{Approx.} & \textbf{Time lower bound} & \textbf{Density} & \textbf{Hypothesis} & \textbf{Ref.}  \\ \midrule
         \checkmark & $\text{poly}(n)$ & \(n^{\omega - o(1)}\) & Dense & SV & \cref{thm:lb:sv} \\
         \checkmark & $\text{poly}(n)$ & \(n^{\omega - o(1)}\) & Dense & HD-OV & \cref{thm:lb:hdov} \\
         \midrule
         x & \(<2\) & \(n^{3 - o(1)}\) & Dense & BMM & \cref{thm:lb:bmm}  \\
         x & \(<3/2\) & \(m^{2-o(1)}\) & Sparse & OV & \cref{thm:lb:ov} \\
         \bottomrule
    \end{tabular}
    \caption{Overview of our lower bounds.}
    \label{tab:lb}
\end{table}

\subsection{Unweighted Graphs}
In unweighted graphs, based on standard hardness assumptions, we cannot obtain small approximation factors more efficiently than the $(1+\smol)$ approximation based on matrix multiplication.

\begin{restatable}{corollary}{corLbBmm}
    \label{thm:lb:bmm}
    Unless combinatorial BMM fails, there is no combinatorial better-than-2 approximation algorithm for \(\tdOp\) running in time \(\BigO(n^{3-\smol})\), even on unweighted DAGs.
    
    Similarly unless BMM fails, there is no better-than-2-approximation for $\tdOp$ running in time $\BigO(n^{\omega-\smol})$. 
\end{restatable}

This can be compared to the simple exact algorithm for unweighted $\tdOp$ in time $\OTilde(n^\omega)$ that binary searches the smallest power of the adjacency matrix that forms the transitive closure.

Interestingly, we can prove a lower bound for a different time/approximation tradeoff, based on the orthogonal vectors conjecture, which implies SETH.  In particular, we show that the $\tdOp$ approximation remains hard even in sparse unweighted DAGs.

\begin{restatable}{theorem}{thmLbOv}
    \label{thm:lb:ov}
    Unless OV fails, there is no better-than-3/2 approximation algorithm for \(\tdOp\) running in time \(\BigO(m^{2-\smol})\), even on unweighted DAGs.
\end{restatable}

Hence, our lower bounds for unweighted graphs are not as strong as the lower bounds for weighted graphs (w.r.t ~approximation factor). In \Cref{sec:open} we explain why it is technically challenging to obtain a stronger lower bound for unweighted graphs. 

On the other hand, we can directly exclude certain lower bounds by giving the following upper bounds.
We first use a standard sampling argument to get the following randomized result. 

\begin{restatable}{theorem}{thmRandomSampling}
    \label{thm:random-sampling}
    There is an algorithm that given a parameter \(k \in \N\),
    outputs an estimate \(\hat D\) such that with high probability, \(\td{G} - k \le \hat D \le \td{G}\).
    It always runs in time $\BigO(nm/k)$.
\end{restatable}

While the additive term in the approximation will become polynomial if we want faster algorithms, as discussed, it seems very challenging to avoid such additive terms with standard techniques.  

Using ideas from the chain-cover-based approach used in the shortcut set literature mentioned above, we can derandomize this algorithm.
The core subroutine is similar to our $\mcc$-based result, but while $\setsize{\mcc}$ is a property of the input graph, here we use a chain-cover where we can control its size.
A smaller chain-cover must then tolerate a larger approximation error, giving the stated trade-off.

\begin{restatable}{theorem}{thmSamplingDerandomized}\label{thm:sampling_derandomized}
    There is a deterministic algorithm that given a DAG $G$ and a parameter $\ell$, outputs an estimate $\hat{D}$, such that $(\td{G} - \ell) / 3 \le \hat{D} \le \td{G}$ and runs in time $\OTilde(nm / \ell + m^{1 + o(1)})$.
\end{restatable}

Seeing these upper bound results, it is natural to ask whether we can improve our lower bound results to show that the trade-off between a $k$-approximation and $\BigO(nm/k)$ time is the best that one can hope for (at least combinatorially).
The following result, obtained by modifying a shortcut set algorithm due to Jambulapati, Liu, and Bernstein \cite{DBLP:conf/focs/LiuJS19}, shows such a lower bound trade-off does not hold. This also establishes a technical connection between shortcut set computation and reachability diameter.

\begin{restatable}{theorem}{thmLjs}
  \label{thm:ljs}
  There is an algorithm computing a $n^{1/2 + o(1)}$-approximation of $\tdOp$ in $\Oish(m)$ time.
\end{restatable}

At its core, the algorithm relies on sampling pivot vertices and then building recursive subinstances based on the reachability relationship between each vertex and all pivots.
This recursion is lossy in the sense that the diameter might be split up into different subinstances (thus the maximum diameter of all subinstances might be less than that of the input graph).
The quality analysis now relies on tuning the recursion to have this bad event happen the fewest times while maintaining near-linear running time. 
\subsection{Open Problems and Technical Challenges}\label{sec:open}

We are left with two important open problems about the reachability diameter. The most natural one is on bridging the upper bound and lower bound gap for unweighted graphs, which we formally state below. Secondly, we discuss whether this problem can be generally reduced to DAGs.

\paragraph{Best Approximation for Unweighted Graphs.}
Our strong lower bounds for the weighted setting indicate that it is unlikely to find fast approximation algorithms there. In particular, by computing $(1+\epsilon)$-APSP (all pairs shortest paths) in time $\OTilde(n^\omega / \smol)$ we can essentially match the lower bounds in \Cref{thm:lb:hdov,thm:lb:sv}.

However, in the unweighted setting, the picture is not complete. The $\tdOp$ of a graph can be computed exactly with fast matrix multiplication and this cannot be improved polynomially due to the reduction to BMM. Furthermore, even better-than-2 approximations are ruled out in faster time.
On the other hand, we show that faster algorithms with polynomial approximations are possible (\Cref{thm:random-sampling} and \Cref{thm:ljs}).
 This motivates the following open question. 

\begin{open} \label{open:apx-unweighted} What is the best approximation/running time tradeoff for $\tdOp$ in unweighted graphs?

In particular, what is the best approximation possible in either linear-time or in better than $O(n^{\omega})$ time?
\end{open}

It is also possible to answer \Cref{open:apx-unweighted} by proving a stronger lower bound, such as the ones shown in \Cref{thm:lb:hdov,thm:lb:sv}. However, such a lower bound requires new techniques, as we run into a challenge that was previously called the ``triangle-inequality barrier''~\cite{Abboud2023DiameterApprox,Karthik2020,Rubinstein2018TriangleBarrier}.

Concretely, to prove a fine-grained lower bound for a $c$-approximation for $\tdOp$ we need to construct a graph for which the YES and NO instances have the following property: In one case they contain short paths of length at most $\ell$ for every reachable pair of vertices. In the other case, there should be a reachable pair of vertices $u,v$ with distance $c\cdot \ell$.

In a straight-forward way we could create a path of length $c\ell$ from $u$ to $v$ but shorter paths would exist in the first case to ensure a smaller distance. However, there must be intermediate vertices on this path as the graph is unweighted. They would introduce distances up to $c\ell - 1$, making this simple approach implausible. Overcoming this barrier to prove stronger hardness results could therefore reveal interesting techniques how to achieve the same for other distance problems.

\paragraph{Reducing General Directed Graphs to DAGs.}

Conceptually, the hardest instance for computing reachability diameter seems to be when the input is a DAG. One explanation is that our lower bound instances are all DAGs. Another intuition is that we can consider the DAG of SCCs, and run the algorithms for classical diameter on each component. However this does not directly lead to a reduction without losing substantially in the approximation.
Our next open problem aims to formalize this intuition that such a reduction exists.

\begin{open}\label{open:reduction-general-to-dag}
    Can we reduce $\tdOp$ in general directed graphs to $\tdOp$ in DAGs?
\end{open}

Such a reduction would immediately allow us to transfer our results based on chain-covers in \Cref{cor:mcc_approx} and \Cref{thm:sampling_derandomized} to general graphs. It would also allow us to restrict ourselves to DAGs in \Cref{open:apx-unweighted}.
Furthermore, all our hardness instances are DAGs, so for example \Cref{thm:lb:hdov,thm:lb:sv} imply that DAGs are as hard as general directed graphs for weighted approximation.

Recently Assadi, Hoppenworth, and Wein~\cite{dag_cover} proved a similar result for distances in directed graphs. They construct $\BigO(\log n)$ DAGs with $\OTilde(m)$ additional total edges (see also \cite{DBLP:journals/corr/abs-2509-23458}). Then, they guarantee that distance $d(u,v)$ in the original graph is poly-logarithmically approximated by the minimum distance from $u$ to $v$ in the collection of DAGs. Hence, for problems like directed approximate shortest paths with decently large approximation factor, we can assume without loss of generality that the given input graph is a DAG.

However, this result does not solve \Cref{open:reduction-general-to-dag}, since any individual DAG could have a much larger reachability diameter than the original graph. Therefore, we need a slightly different guarantee, and it would be interesting if it is possible to transfer the concept of DAG covers to reachability diameter approximation.

\iflong

\section{Preliminaries}\label{sec:prelims}
For \(n \in \N\), write $[n] \coloneqq \{1,\dots, n\}$.
For a vertex $v$ in some graph, we denote its neighborhood as $N(v)$.
In directed graphs, we write its out- and in-neighborhoods as $N^+(v)$ and $N^-(v)$, respectively.
Note that then $N(v) = N^+(v) \cup N^-(v)$.
The \emph{transitive closure} of a graph $\tc(G)$ is defined as the graph $(V(G), \{(u,v) \in V(G)^2 \mid u \text{ can reach } v \text{ in } G\})$.
A \emph{chain} is a path in the transitive closure of a graph.
In a directed graph $G$, a \emph{chain} is a sequence of vertices $v_1, \dots, v_\ell$ such that for all $i \in [n-1]$, we have $v_i \to v_{i+1} \in \tc(G)$.
Each DAG $G$ encodes a poset $\le_G$, that is, for vertices $u,v$ we write $u \le_G v$ if there is a path from $u$ to $v$ in $G$.
Where it is clear from context, we omit the subscript in $\le_G$.

The \emph{width} of a DAG is the size of its minimum chain cover, that is the minimum number of chains required to cover every vertex.
It is a common graph parameter used to study the complexity of problems on DAGs \cite{DBLP:conf/soda/KoganP23,DBLP:conf/esa/KoganP24,caceres_minimum_2023,DBLP:conf/esa/KoganP24a,DBLP:conf/soda/Obdrzalek06}.
Equivalently it can be defined as the size of the largest independent set in the transitive closure (these are also known as \emph{anti-chains}).
The \emph{treewidth} of an undirected graph is the minimum largest bag size of any tree-decomposition.
A tree decomposition of a graph $G$ is a tree $T$ where every vertex $X \in V(T)$ is associated with a subset of the vertices in $G$.
We call these sets \emph{bags} and treat the tree-vertices directly as their bags.
The following three properties hold: (1) for every $v \in V(G)$, there is a bag $X \in V(T)$ s.t.\,$v \in X$, (2) for every edge $\{u,v\} \in E(G)$, there is a bag $X \in V(T)$ s.t.\,$u,v \in X$, and (3) for all $v \in V(G)$, the bags of $T$ that include $v$ form a (connected) subtree of $T$.
Graph families with bounded treewidth include cactii, series-parallel graphs, and outerplanar graphs.
For a more complete discussion of treewidth, see \cite{cygan2015parameterized}.
While there are many attempted definitions of the treewidth of a directed graph, we will only use it to mean the treewidth of the underlying undirected graph \ie{the graph obtained by forgetting all edge directions}.
A \emph{$c$-balanced separator} in a graph $G$ is a set $S$ such that each connected component in $G - S$ has at most $c \setsize{V(G)}$ vertices.
We omit $c$ if it is some constant in $(0,1)$ that is independent of the graph and its size.

\subsection{Hypotheses for Conditional Lower Bounds}\label{sec:hypotheses}

We obtain our lower bound results by fine-grained reductions from more or less standard hypotheses from fine-grained complexity.
We dedicate this subsection to their definition as well as some context for each hypothesis

\paragraph{Boolean Matrix Multiplication.}
The combinatorial boolean matrix multiplication (BMM) hypothesis reads as follows.

\begin{hypothesis}[Combinatorial BMM]\label{hyp:bmm}
    Given vectors $A, B \in \{0,1\}^{n \times n}$, there is no combinatorial\footnote{``Combinatorial'' algorithms generally informally describe algorithms that avoid the use of fast matrix multiplication techniques. See for example~\cite{combinatorial_bmm_discussion} for a detailed discussion of the term and motivations behind combinatorial algorithms and lower bounds.} algorithm that computes $C \in \{0,1\}^{n \times n}$ with $C_{ij} = \bigvee_k A_{ik} \wedge B_{kj}$ in time $\BigO(n^{3 - \smol})$, for any $\smol > 0$.
\end{hypothesis}

Algebraically, that is, using fast matrix multiplication, combinatorial BMM can clearly be solved in time $\BigO(n^\omega)$ by a single matrix multiplication for $\omega < 2.372$~\cite{current_omega_value}. This is also the best general algorithm for BMM and it is generally assumed that no better algorithm exists~\cite{bergamaschi2021new,dalirrooyfard_approximation_2022}. This is why we also use BMM to prove conditional $n^\omega$ lower bounds.
Restricting our attention to combinatorial algorithms, a truly subcubic algorithm for triangle detection in tripartite graphs would imply the same for BMM~\cite{bmm_triangle_detection}. Hence, we use that problem for hardness results under \Cref{hyp:bmm}.

\paragraph{Orthogonal Vectors.}

The orthogonal vectors (OV) hypothesis is one of the most important hypotheses in fine-grained complexity as is was used to prove many lower bounds (see~\cite{fine_grained_survey} for a summary) and follows from the strong exponential time hypothesis~\cite{seth_implies_ov}. Most importantly for our case, it is a well-suited hypothesis for diameter lower bounds and has been used in different variants for standard diameter~\cite{first_diam_lower_bound,another_diam_lower_bound,bonnet_diam_hardness,li_diam_hardness} and $\MinDiam$~\cite{abboud_approximation_2016}. Our OV-based hardness result follows these previous hardness constructions.

\begin{hypothesis}[OV]
    Given sets $A, B$ of size $n$ of vectors in $\{0,1\}^d$, with $d \ge \omega(\log n)$ there is no algorithm that decides if there is $a \in A$ and $b \in B$ with $\sum_i a_i b_i = 0$ in time $\BigO(n^{2 - \smol})$, for any $\smol > 0$.
\end{hypothesis}

\paragraph{High-Dimensional OV.}

The high-dimensional OV hypothesis (HD-OV) considers the OV problem for higher dimensions, that is, we now have $d = \Theta(n)$. It can be thought of as an analogue of the OV hypothesis for dense graphs. In this setting, fast matrix multiplication produces a running time of $\BigO(n^\omega)$, a speedup over the brute-force-way to solve standard OV. 

This problem was previously used as the basis of a lower bound for $\MinDiam$ by Dalirrooyfard and Kaufmann~\cite{dalirrooyfardApproximationAlgorithmsMinDistance2022}. They showed that a better-than-$3/2$-approximation for $\MinDiam$ requires time $\Omega(n^{\omega - o(1)})$. 
\begin{hypothesis}[HD-OV]
    Given sets $A, B$ of size $n$ of vectors in $\{0,1\}^d$ with $d = \Theta(n)$, there is no algorithm that decides if there is $a \in A$ and $b \in B$ with $\sum_i a_i b_i = 0$ in time $\BigO(n^{\omega - \smol})$, for any $\smol > 0$.
\end{hypothesis}

\paragraph{Simplicial Vertex.}

The simplicial vertex problem (SV) is a graph problem that asks you to decide if there is a vertex whose neighborhood is a clique. It can again be solved in fast matrix multiplication time~\cite{sv_algorithm}.
The problem has been used as a basis for lower bounds for the clique cutset problem in~\cite{kratsch_simplicial_lb}, for lower bounds related to $n$-pairs diameter~\cite{dalirrooyfard_approximation_2022}, and for undirected diameter in structured graphs~\cite{ducoffe2022diameter}. 

\begin{hypothesis}[SV]
    Given an undirected, unweighted graph $G$, there is no algorithm that decides if there exists a vertex $v \in V(G)$ such that $N(v)$ is a clique in time $\BigO(n^{\omega - \smol})$, for any $\smol > 0$.
\end{hypothesis}

The usefulness of this hypothesis stems from two main properties. First is its quantifier structure of $\exists \forall$ or $\forall \exists$ when negated, which is the same structure that is necessary for upper bounding the diameter (for all vertex pairs, there is a short path). Secondly, the first quantifier quantifies over a set of size $n$ which is necessary to avoid a blow-up in the reduction in our case.

\section{Weighted Graphs}
\label{sec:weighted}

In this section we prove our results for weighted graphs. We start with strong lower bound for any approximation algorithm for $\tdOp$ in \Cref{sec:weighted-lb}. Our positive results are hence restricted to special graph classes. In \Cref{sec:chain-based}, we give an algorithm for DAGs of small width.
In \Cref{sec:treewidth}, we go on to give an algorithm for small-treewidth graphs and highlight connections to shortcut sets and hopsets.

\subsection{Lower Bounds}\label{sec:weighted-lb}

First, we prove our lower bound based on the simplicial vertex hypothesis. We use the hypothesis in a similar way to~\cite{dalirrooyfard_approximation_2022,ducoffe2022diameter}. See \Cref{sec:hypotheses} for a definition and discussion of hypotheses.

\begin{figure}[tp]
    \centering
    \begin{minipage}[t]{.48\linewidth}
        \centering
        \includegraphics*[page=3]{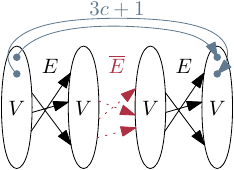}
        \subcaption{Lower bound construction from \Cref{thm:lb:sv}. The graph has $\tdOp$ 3 iff no vertex is simplicial, that is, the neighborhood of no vertex is a clique.}
        \label{fig:sv-lb-construction}
    \end{minipage}
    \hfill
    \begin{minipage}[t]{.46\linewidth}
        \centering
        \includegraphics*[page=4]{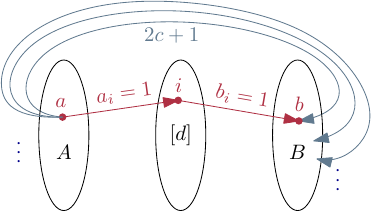}
        \subcaption{Lower bound construction from \Cref{thm:lb:hdov}. The graph has $\tdOp$ 2 iff there is no orthogonal pair.}
        \label{fig:hdov-lb-construction}
    \end{minipage}
    \caption{Our lower bound constructions for weighted $\tdOp$.}
\end{figure}

\thmLbSV*

\begin{proof}
    We show that such an approximation algorithm would imply the simplicial vertex hypothesis is false.
    Let $c = \text{poly}(n)$ be the approximation factor of that algorithm and let $G$ be an undirected input graph we are to determine the existence of a simplicial vertex in.
    
    As input to our $\tdOp$ approximation algorithm, create a graph $H$ as follows.
    \begin{itemize}
        \item  Let $V_1, V_2, V_3, V_4$ be four disjoint copies of $V(G)$ and set $V(H) \coloneqq V_1 \sqcup \dots \sqcup V_4$.
        \item For any edge $\{u,v\} \in E(G)$, create edges $v_1 \to u_2$, $u_1 \to v_2$, $v_3 \to u_4$, and $u_3 \to v_4$, where the indices indicate the copies in the sets $V_1$ to $V_4$.
        \item For any non-edge $\{u,v\} \not \in E(G)$, create edges $v_2 \to u_3$, and $u_2 \to v_3$.
    All these edges have weight 1.
        \item Additionally, for every $v \in V(G)$, create an edge $v_1 \to v_4$ with weight $3c + 1$.
    \end{itemize}

    We prove that this graph has a $\tdOp$ of $3$ if there is no simplicial vertex and $3c+1$ otherwise.
    Therefore a $\BigO(n^{\omega - \smol})$ $c$-approximation yields a $\BigO(n^{\omega - \smol})$ algorithm to decide the existence of a simplicial vertex.

    First, observe that for any pair of vertices in $H$ that are not the $V_1$ and $V_4$ copy of the same vertex in $G$, either no connection exists or the connection only uses at most three edges of weight 1.
    Thus, these pairs are either not reachable or have distance at most $3$.
    Therefore, we can ignore them for the rest of the analysis and focus on the distance between the $V_1$ and $V_4$ copies of the same vertex from $V(G)$.

    Now, assume that $G$ has a simplicial vertex $v$.
    Our goal is to show that the path consisting of the single edge $v_1 \to v_4$ is a shortest path with length $3c+1$ and thus $\td{H} \ge 3c+1$.
    Assume for the sake of contradiction, that there is a shorter path from $v_1$ to $v_4$.
    Then, since $H$ is a layered graph with 4 layers (except the direct edges between $V_1$ and $V_4$), this path must be of the form $v_1 \to a_2 \to b_3 \to v_4$, where $a_2 \in V_2$ and $b_3 \in V_3$.
    By construction of $H$, the edges $v_1 \to a_2$ and $b_3 \to v_4$ mean that $a,b \in N_G(v)$.
    The edge $a_2 \to b_3$ means that $\{a,b\}$ is \emph{not} an edge in $G$.
    Therefore, $v$ is not a simplicial vertex in $G$, contradicting our assumption and we can conclude that $d_H(v_1,v_4) = 3c+1$ and $\td{H} \ge 3c+1$, as desired.

    Now, assume that $G$ has no simplicial vertex.
    Let $v$ be an arbitrary vertex.
    Because $v$ is not a simplicial vertex, there are vertices $a,b \in N(v)$ such that $\{a,b\}$ is not an edge in $G$.
    Therefore $v_1 \to a_2 \to b_3 \to v_4$ is a path of length 3 in $H$ and thus $d_H(v_1, v_4) \le 3$.
    Since we chose $v$ arbitrarily, we have that $\td{H} \le 3$.
\end{proof}

Our next lower bound proves the same bound but based on the high-dimensional OV hypothesis instead. It hence strengthens the previous theorem.

\thmLbHdov*

\begin{proof}
    We use the standard OV graph construction to reduce high-dimensional OV to $\tdOp$ approximation. Assume there is an algorithm with approximation factor $c \in \mathbb{R}^+$ for $\tdOp$ in time $\BigO(n^{\omega - \smol})$ for some $\smol \in (0, \omega - 2]$. See \Cref{fig:hdov-lb-construction} for a drawing of the reduction.
    
    Let $A,B \subseteq \{0, 1\}^d$, with $|A|, |B| = n$ and $d \in \Theta(n)$. Create a graph as follows. \begin{itemize}
        \item  Create $V \coloneqq A \sqcup B \sqcup [d]$.
        \item For $a \in A$ and $i \in[d]$, create an edge $a \to i$ iff $a_i = 1$.
        \item Similarly, for $b \in B$, create an edge $i \to b$ iff $b_i = 1$.
        \item All of these edges have weight 1.
        \item Additionally, for all $a \in A, b \in B$ add an edge $a \to b$ of weight $2c+1$.
    \end{itemize}
    
    Note that between any two $a,b$ there now is a path of length 2 via $[d]$ if and only if $a$ and $b$ are \emph{not} orthogonal.
    Because of the last edges, an orthogonal pair of vectors has distance $2c+1$. On the other hand, $\tdOp > 2$ also implies that there is an orthogonal pair of vectors.
   
    Since the graph has $\BigO(n)$ vertices, $\BigO(n^2)$ edges, and $\smol \le \omega - 2$, the running time of the approximation algorithm dominates.
    Hence, a $c$-approximation in time $\BigO(n^{\omega - \smol})$ can also decide high-dimensional OV in the same time.
\end{proof}

\subsection{Chain-Based Algorithm}\label{sec:chain-based}
The width of a DAG $G$ gives a decomposition of $G$ into chains.
In the introduction, we discussed how the $\tdOp$ of a DAG is equivalent with its $\MinDiam$ if the DAG consists of a single path plus extra edges along that path.
This is equivalent to having width one.
Thus, we can use algorithms for $\MinDiam$ to compute a constant-factor approximation of $\tdOp$ in subcubic time for DAGs of width one.
This raises the natural question of whether we can extend this idea to larger widths.
Recall that the width is the size of the minimum chain cover.
The rough strategy is thus as follows:

\begin{enumerate}
    \item  Decompose $G$ into chains.
    \item Use $\MinDiam$-inspired techniques to approximate $\tdOp$ between vertices on the same chain (will be \Cref{lem:chain-2-approx}).
    \item  Use additional ideas to approximate $\tdOp$ between vertex pairs on different chains (will be \Cref{lem:chain-bfs}).
    \item Combine this information (will be \Cref{thm:chain_3_approx}).
\end{enumerate}

In this section, we will see how to implement this strategy.
We can achieve step (1) using a result by Cáceres.
\begin{theorem}[\cite{caceres_minimum_2023}]
    There is an $\BigO(m^{1+o(1)})$ time algorithm that computes a minimum chain cover of a DAG. 
    \label{thm:mcc-alg}
\end{theorem}

So, let us now look at step (2) of our strategy. The following lemma extends Theorem 2.2 from \cite{abboud_approximation_2016}.
We show that with a slight extension their algorithm can be used on any chain inside a DAG to approximate the maximum distance between any two vertices on the chain.
Note that the shortest path between (even sequential) vertices on the chain might use intermediate vertices not on the chain, thus requiring us to go slightly beyond the original algorithm by \cite{abboud_approximation_2016}.

\begin{algorithm}
    \textbf{Input:} DAG $G$, chain $C = v_1, \dots, v_\ell$, numbered in topological order
    
    \textbf{Output:} 2-approximation for the maximum distance of two vertices in $C$
    
    \begin{enumerate}%[noitemsep, parsep=0pt, topsep=0pt]
        \item Compute a topological order $O$.
        \item Remove all vertices before $v_1$ or after $v_\ell$ in $O$.
        \item Let $V_1, V_2$ be a partition of $V(G)$ such that $V_1 = \lfloor \setsize{V} / 2 \rfloor$ and no vertex from $V_1$ appears after any vertex from $V_2$ in $O$. Let $v_k$ be the last vertex from $C$ in $V_1$.
        \item Run Dijkstra's algorithm forwards and backwards from $v_k$.
        \item Run the algorithm recursively on $G[V_1]$ with chain $v_1, \dots, v_k$ and $G[V_2]$ with chain $v_{k+1}, \dots, v_\ell$.
    \end{enumerate}%

    Report the depth of the deepest vertex from the chain in any shortest path tree seen during the execution. 
    
    \caption{2-approximation of intra-chain diameter.}%
    \label{algo:dag_chain_2_approx}%
\end{algorithm}

\begin{lemma}
    \label{lem:chain-2-approx}
    Given a DAG $G$ and a chain $C = v_1, \dots, v_\ell$, \Cref{algo:dag_chain_2_approx} computes a 2-approximation of the maximum distance of any two sequential vertices on the chain in $\OTilde(m)$ time.
\end{lemma}

\begin{proof}
    For the running time analysis, observe that this divide and conquer algorithm generates two subproblems, each with at most half as many vertices.
    In each recursive call, the additional overhead (a topological sort, a forwards and backwards Dijkstra) runs in \(\BigO(n + m)\).
    Solving this recurrence leaves us with a total running time of \(\BigO(m \log n)\).

    To prove correctness, assume that $a$ and $b$ are two vertices on the chain with maximum pair-wise distance.
    Assume wlog. that $a$ appears before $b$ on the chain (and thus in the topological order).
    In each recursion step, $v_k$ can be before both $a$ and $b$, after both, or between them.
    It the first two cases, $a$ and $b$ are in the same recursive subproblem, so let us focus on the case where $a \le v_k \le b$ for the first time.
    Let $V'$ be the vertex set of this recursive call.
    Since $a,b \in V'$, note that $V'$ must contain the entire subchain from $a$ to $b$.
    Similarly, since the divide-steps respect the topological order, any vertex that is on a path between $a$ and $b$ must be in $V'$.
    Thus $d_G(a,b) = d_{G[V']}(a,b)$.
    By, the triangle inequality, we have $d(a,b) \le d(a,v_k) + d(v_k,b)$.
    By this and the choice of $a$ and $b$, we can conclude that
    \[
    \max(d(a,v_k), d(v_k,b)) \le d(a,b) = \max_{u \le v \in C} d(u,v) \le 2 \max(d(a,v_k), d(v_k,b)).
    \]
    By the definition of the output of the algorithm, this shows the desired property.

    Clearly, any distance that the algorithm outputs corresponds to the length of a path between two vertices on the chain. Therefore, the algorithm can never overestimate the diameter. 
\end{proof}

Next, we show how to compute an estimate for $\tdOp$ with the starting vertex on the chain.
This is step~(3) of our plan.

\begin{lemma}
    \label{lem:chain-bfs}
    There is an algorithm that given a directed graph $G$ and a chain $C = v_1, \dots, v_\ell$, computes an estimate $\hat D$ such that $\hat D \le \max_{a \in C, b \in V(G)} d(a,b) \le \hat D + \max_{a,b \in C} d(a,b)$ in $\Oish(m)$ time.
\end{lemma}

\begin{proof}
    We start by computing the shortest path trees $T_i \coloneqq \text{\textsc{Dijkstra}}(v_i, G \setminus \bigcup_{j=i+1}^\ell T_i)$ with $T_\ell = \text{\textsc{Dijkstra}}(v_\ell, G)$. By working backwards, starting from $i = \ell$, this is possible in time $\Oish(m)$. We simply keep track of the current tree and already visited vertices by other trees. So each edge will be visited at most once.
    Then, we report the height of the highest tree as the estimate $\hat D$.

    We prove the two inequalities on $\hat D$ separately and start with $\hat D \le \max_{a \in C, b \in V(G)} d(u,v)$.
    Towards this goal, assume $T_i$ is one of the highest trees, and $b$ is a vertex in its last layer.
    Recall that $v_i$ is the root of $T_i$.
    To show the inequality, we show that the path from $v_i$ to $b$ in $T_i$ is a shortest path.
    Assume by contradiction that it is not.
    Then the other path that is shorter must contain a vertex $u$ that is in a $T_j$ with $j>i$ (because otherwise $u$ would be in $T_i$).
    But then since $u$ can reach $b$ and $v_j$ can reach $u$, we also have that $v_j$ can reach $b$ and therefore $b$ must be in one of the trees $T_j, \dots, T_\ell$, contradicting our assumption.

    Let us now prove that $\max_{a \in C, b \in V(G)}d(a,b) \le \hat D +c$, with $c = \max_{a,b \in C} d(a,b)$. Let $a \in C, b \in V(G)$ be one of the pairs maximizing $d(a,b)$.
    Assume $a = v_i$ and let $T_j$ be the tree such that $b \in T_j$.
    Then $a \rightsquigarrow v_j \rightsquigarrow b$ is a path of length at most $c + \hat D$ and therefore $d(a,b) \le \hat D + c$, as desired.
\end{proof}

Now, we combine the two previous lemmas. We compute estimates for both the diameter on the chain as well as the diameter leaving the chain. These can then be combined to give a 3-approximation for any path with one endpoint on the chain. This is step (4).

\begin{theorem}\label{thm:chain_3_approx}
    There is an algorithm that given a DAG $G$ and a chain $C$ in $G$ computes an estimate $\hat{D}$ with $\hat{D} \le \max_{a \in C, b \in V(G)} d(a,b) \le 3\hat{D}$ in time $\OTilde(m)$.
\end{theorem}
\begin{proof}
    We run the algorithm from \Cref{lem:chain-bfs} to get an estimate $\hat{D}$, and the algorithm from \Cref{lem:chain-2-approx}. From the set of shortest path trees, we take the maximum height $h$ of any tree and return the maximum $\hat{D}^* \coloneqq \max \{h, \hat{D}\}$.

    The runtime is clearly as claimed. By \Cref{lem:chain-2-approx}, $\hat{D}$ is always bounded by some distance on the chain. Furthermore, $h$ is a valid distance, since anything reached in a shortest path tree is not in a shortest path tree from a later vertex on the chain. Hence, the first inequality holds. On the other hand, we have \[\max_{a \in C, b \in V(G)} d(a,b) \le \max_{a,b \in C}d(a,b) + h \le \hat{D}^* + 2\hat{D}^* = 3\hat{D}^*.\qedhere\]
\end{proof}

Computing a minimum chain cover using \Cref{thm:mcc-alg} and then applying \Cref{thm:chain_3_approx} to every chain in it gives the following result.

\thmMccApprox*

\subsection{Treewidth-Based Algorithm}
\label{sec:treewidth}
In this section, we take a structural perspective and explore the relationship between diameter estimation and treewidth.
We also deepen the connection between shortcut sets, hopsets and diameter estimation under this lens.
First, we briefly explain how an algorithm from~\cite{abboud_approximation_2016} can be slightly modified to exactly compute the $\tdOp$ of a directed graph whose underlying undirected graph has bounded treewidth.
We then give a simplified version of that algorithm that yields as 2-approximation and has a polynomial dependence on the treewidth (compared to the exponential dependence in
\cite{abboud_approximation_2016}).
Their approach relies on a complicated data structure that we replace with a simple, explicit construction to obtain a shortcut sets. We then show how this leads to exact hopsets with stretch $2$ of size $\Oish(n \cdot f(\tw))$.

The core fact exploited by all these techniques is that an undirected graph of bounded treewidth has a small bounded separator.
Since this holds recursively, we can split the graph in a balanced way, while handling all paths that touch the bounded-size separator.
Interestingly, this separator of the underlying undirected graph is useful for computing the (reachability) diameter in directed graphs and for TC-spanners, shortcut sets and hopsets.
Formally, we will use the following lemma.

\begin{lemma}[Lemma 7.19 of \cite{cygan2015parameterized}]
\label{lem:balanced-sep}
    Let $T$ be a tree decomposition of an undirected graph $G$ of bag size $k$.
    Then there is a $1/2$-balanced separator of $G$ of size $k$.
    It can be computed from $T$ in linear time.
\end{lemma}

\begin{proofsketch}
   Pick an arbitrary root $R$ for the tree decomposition.
   Then, for each bag $X$ compute the number of vertices from $G$ in all bags below $X$ using a simple DP.
   Finally, as the separator, pick a bag that (a) has at least $V(G)/2$ vertices in bags below it, but (b) none of its child-bags do.
\end{proofsketch}

\subsubsection{Exact ReachDiam}

Let us start by remarking how to adapt the algorithm in \cite{abboud_approximation_2016} to compute the reachability diameter.
Unfortunately, we cannot state the result to be self-contained here without replicating much of their work, we thus settle for explaining the required modification.

\begin{observation}[Cf. Theorem 3.3 of \cite{abboud_approximation_2016}]
    There is an algorithm that given a directed graph $G$ of treewidth $\mathrm{tw}$, computes $\td{G}$ in time $\OTilde(m \cdot \tw^2 \log^{\tw-1}(n))$.
\end{observation}

\begin{proofsketch}
  As their algorithm computes the eccentricities of each vertex $u$ in the graph as $\max_{v \in V(G)} d(u,v)$, we need to adapt it such that only finite distances are considered in this maximum. 
  This can be achieved by only inserting finite values in to the range searching for maximum data structure they use internally.
  It is easy to verify that this does not affect running time and that the correctness stays essentially the same, except for the eccentricity definition.
\end{proofsketch}

\subsubsection{A Simpler 2-Approximation}
Now let us take closer look at the algorithmic strategy of how to use treewidth by giving the simpler 2-approximation explicitly. Note that the linear dependence on treewidth is only reached by using recent fast max flow subroutines. Using different treewidth approximations from~\cite{dong2024min-cost} a running time of $\OTilde(m \cdot \mathrm{tw}^3)$ would also be achieveable, which is better for polylogarithmic treewidth.

\begin{theorem}\label{thm:treewidth}
    There is an algorithm that given a directed graph $G$ of treewidth $\mathrm{tw}$, computes a 2-approximation of $\td{G}$ in time $\OTilde(m \cdot \mathrm{tw} + m^{1 + o(1)})$. 
\end{theorem}

\begin{algorithm}
    \textbf{Input:} Directed graph $G= (V,E)$

    \textbf{Output:} 2-approximation of $\tdOp$
    
    \begin{enumerate}%[noitemsep, parsep=0pt, topsep=0pt]
        \item Compute a tree decomposition $T$ of $G$ with bag size at most $\tw \cdot \log^3(n)$\;
        \item Pick one of the bags $S$ that is a balanced separator.\label{step:balanced-separator}
        \item Run a forwards and backwards Dijkstra from every $s \in S$ \label{step:dijkstra}.
        \item For every $s, s' \in S$, add an edge $s \to s'$ of weight $d(s,s')$ if that distance is finite.
        \item For every component $C$ of the undirected graph $G - S$, recurse on the directed graph $G[C \cup S]$.
    \end{enumerate}%
    
    Return the maximum distance seen in any Dijkstra (also from the recursive instances).
    \caption{Approximation algorithm for the $\tdOp$ of a graph of bounded treewidth.}%
    \label{algo:tw-2approx}%
\end{algorithm}

\begin{proof}
   We consider \Cref{algo:tw-2approx}.

    For the running time, we can compute a tree decomposition once in the beginning using recent fast tree width approximations in time $m^{1 + o(1)}$~\cite{Bernstein2022}. The algorithm to find a balanced separator (\Cref{lem:balanced-sep}) runs in linear time and guarantees that the separator contains only vertices from a single bag in the tree decomposition. Since that separator will correspond to a set of vertices in the same bag, adding edges between them does not increase the bag size of our tree decomposition. 
    
    To analyze the recursion, first notice that there are $\BigO(\log n)$ levels until the size of the graph depends only on $\mathrm{tw}\cdot \log n$. At that point, we can compute the diameter in cubic time. On every level of the recursion, we touch each original edge at most $\OTilde(\mathrm{tw})$ times, once from each Dijkstra. The newly added edges are at most $\mathrm{tw}^2\log n$ in every subinstance and thus do not further increase the running time. In total, we arrive at the claimed running time of $\OTilde(m \cdot \mathrm{tw})$. 

    Regarding the correctness, first notice that distances let $d(a,b) = \td{G}$ and consider the recursive step where either at least one of $a$ and $b$ is in $S$, or $a$ and $b$ are on different sides of $S$. In the first case $\hat{D} = d(a,b)$ in this step. In the second case, there is a vertex $s \in S$ that lies on the shortest path from $a$ to $b$. Hence, in that case either $d(a,s)$ or $d(s,b)$ is a 2-approximation.

    If $a$ and $b$ lie in the same component of $G-S$, we have to argue that $d_{G_C}(a,b) = d_G(a,b)$. If the previously shortest path was completely contained in $C$, it remains the shortest path, since the additional edges correspond to previous distances and thus cannot add shorter paths. If the previously shortest path left the component $C$, then it must first lie in $C$, then use a vertex $s_1 \in S$, and at the end use another vertex $s_2 \in S$ and afterwards stay in $C$. The subinstance $G_C$ contains a direct edge from $s_1$ to $s_2$ with weight $d_{G}(s_1, s_2)$, hence a path with the same distance still exists.
\end{proof}

\subsubsection{Adaptation to Shortcut Sets and Hopsets}
Let us now see how the previous algorithm can be adapted to yield shortcut sets and hopsets. Explicitly, we add the shortcut edges (or weighted hopset edges corresponding to distance of the endpoints) from and to the balanced separator in Step \ref{step:dijkstra}.

We focus our presentation on hopsets as the statement for hopsets directly implies the shortcut set version and as the shortcut set version, while it hasn't been explicitly stated, is implied by the result on TC-spanners in $H$-minor-free graphs in \cite{DBLP:journals/siamcomp/BhattacharyyaGJRW12}.

Recall that in a (non-negatively weighted) directed graph $G$, a $(\beta, \smol)$-hopset is a set $H \subset V(G) \times V(G) \times \R$ of new (weighted) edges such that (1) no distance decreases from $G$ to $G \cup H$ (including, becoming finite) and (2) in $G \cup H$ for all $u,v$ with $d(u,v) < \infty$, there is a path in $G \cup H$ of $\beta$ vertices of length at most $(1+\smol) d(u,v)$.
If $\smol=0$, we have an exact $\beta$-hopset.

We note that the special case of $\beta$-hopsets for $\beta = 2$ were also studied in the context of \emph{hub labelings}. In that context, similar techniques as in \Cref{thm:tw-hopset} were used (see~\cite{ducoffe2022eccentricity,gavoille2004distance}). 

\thmTwHopset

\begin{proof}
  During the execution of the algorithm from \Cref{thm:treewidth}, for each recursive instance on $V'$ with separator $S$ do the following.
  For each $s \in S$, add edges from $s$ to all vertices $s$ can reach with the distance as the weight to the hopset $H$.
  Similarly, add edges to $s$ from all vertices that can reach $s$ with the distance as the weight to $H$.
  For the running time, it is easy to see that this information is revealed from the Dijkstra calls performed in the algorithm.
  For the hopset size, use the same recurrence, with cost $\setsize{S} \times \setsize{V'}$ for each recursive instance.
  For the hopbound, let $u,v$ be vertices with distance $d(u,v) < \infty$. 
  At the beginning $u,v$ are in the same subinstance (since there is only one).
  Then, in each recursive step, one of the following must happen.
  In case 1, at least one of $u,v$ ends up in the separator, then $H$ will include direct edge of weight $d(u,v)$.
  In case 2, $u,v$ end up in the same subinstance, then recall from the proof of \Cref{thm:treewidth} that the distance between them does not change in this subinstance.
  In case 3, $u,v$ are separated by the separator $S$.
  Then there is a $s \in S$ that lies on a shortest $u,v$ path, thus edges $u \to s$ and $s \to v$ will be added such that $d(u,v) = d(u,s) + d(s,v)$.
  As at some point of the recursion, case 1 or 3 must occur, there is a shortest path with 2 hops between $u$ and $v$ in $G+H$, showing that $H$ is an exact 2-hopset.
\end{proof}

\section{Unweighted Graphs}
\label{sec:unweighted}

In this section, we prove our results about unweighted graphs. First, we rule out small constant-factor approximations in \Cref{sec:unweighted-lb}. Next, we show how to use our algorithm for weighted DAGs with small width to derandomize a classical randomized sampling algorithm in \Cref{sec:sampling}. Finally, we provide our best approximation guarantee in near-linear time that achieves an $n^{1/2 + o(1)}$-approximation in \Cref{sec:fineman}, showing that we can go beyond a natural trade-off between approximation guarantee and running time suggested by the classical sampling algorithm.

This further deepens the connection between $\tdOp$ and shortcut sets and hopsets, as in the previous subsections we have seen a number of algorithms for approximating $\tdOp$ inspired by shortcut set algorithms and we now see the other direction: a diameter estimation algorithm yielding a shortcut set algorithm.
Similar ideas also underlie the construction of almost-linear size TC-spanners for $H$-minor-free graphs in \cite{DBLP:journals/siamcomp/BhattacharyyaGJRW12} (their result implies a comparable result for shortcut sets but not for hopsets).

\subsection{Lower Bounds}\label{sec:unweighted-lb}

\begin{figure}[tp]
    \centering
    \begin{minipage}[t]{.48\linewidth}
        \centering
        \includegraphics*[page=1]{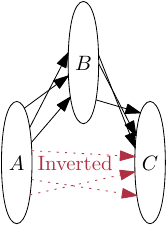}
        \subcaption{Lower bound construction from \Cref{thm:lb:bmm-1-vs-2}. This graph has \(\tdOp\) 2 iff. there is a triangle in the original graph.}
        \label{fig:bmm-lb-construction}
    \end{minipage}
    \hfill
    \begin{minipage}[t]{.46\linewidth}
        \centering
        \includegraphics*[page=2]{diameter-drawings.pdf}
        \subcaption{Lower bound construction from \Cref{thm:lb:ov}. This graph has \(\tdOp\) 3 iff. there are orthogonal vectors in $a \in A,b \in B$.}
        \label{fig:ov-lb-construction}
    \end{minipage}
    \caption{Our lower bound constructions for unweighted graphs.}
\end{figure}

First, we give a strong lower bound under the combinatorial BMM conjecture.

\begin{theorem}
    Unless combinatorial BMM fails, there is no combinatorial algorithm that differentiates between \(\tdOp\) 1 and 2 in DAGs running in time \(\BigO(n^{3 - \smol})\). 
    \label{thm:lb:bmm-1-vs-2}
\end{theorem}

\begin{proof}
    We reduce from triangle detection in a tripartite graph $G$ with partitions $A,B,C$. Create $G'$, by directing the edges from $A$ to $B$, from $B$ to $C$.
    Take the complement of the edges between $A$ and $C$ and direct them from $A$ to $C$.
    See \Cref{fig:bmm-lb-construction} for an illustration.

    If there is a triangle $a,b,c$ in $G$, there is a path $a \to b \to c$ in $G'$ but no $a \to c$ edge.
    Hence, the diameter of $G'$ must be $2$, since this is a shortest path of length 2 and by construction, there are no longer paths.
    If the diameter of $G'$ is 2, there must be a pair $a,c$ with a path $a\to b \to c$ but no $a \to c$ edge. This triple corresponds to a triangle in $G$.

    Thus, a combinatorial algorithm in time $\BigO(n^{3-\smol})$ translates into an algorithm for triangle detection in the same runtime, which is impossible under the combinatorial BMM hypothesis.
\end{proof}

This immediately gives the following result.

\corLbBmm*

Using OV, we show a different tradeoff between time and approximation factor.
This result is not only more meaningful for sparse graphs, it is also based on a different, arguably more trusted, hypothesis. 

\thmLbOv*

\begin{proof}
    We start with the standard embedding of an OV instance as a graph. Let $A, B \subseteq \{0,1\}^d$ with $\setsize{A}, \setsize{B} \in \Theta(n)$.
    \begin{itemize}
        \item  Create $V \coloneqq A \sqcup B \sqcup [d]$.
        \item For $a \in A$ and $i \in[d]$, create an edge $a \to i$ iff. $a_i = 1$.
        \item Similarly, for $b \in B$, create an edge $i \to b$ iff. $b_i = 1$.
        \item To this graph, add two vertices $x$ and $y$.
        \item Add edges from all vertices in $A$ to $x$, from $x$ to $y$, and from $y$ to all vertices in $B$.
    \end{itemize}
    
    See \Cref{fig:ov-lb-construction} for an illustration.
    
    Note that between any two $a,b$ there is a path of length 2 via $[d]$ if and only if $a$ and $b$ are \emph{not} orthogonal.
    Additionally, there is a path of length 3 from every vertex in $A$ to every vertex in $B$ via $x$ and $y$.
    It is easy to verify that this graph has a $\tdOp$ of 3 if there is an orthogonal pair in $A \times B$ and a $\tdOp$ of 2 otherwise.

    Notice that the graph contains $\BigO(n)$ vertices and $\OTilde(n)$ edges.
    Therefore, under the orthogonal vectors hypothesis, there can be no \(\BigO(m^{2-\smol})\) time algorithm that distinguishes between \(\tdOp\) of 2 and 3, where $m$ is the number of edges in the graph. 
\end{proof}

\subsection{Sampling-Based Algorithm and Its Derandomization}\label{sec:sampling}

A simple folklore technique to approximate the diameter is based on sampling. After sampling a set $S$ of size $\OTilde(n / s)$ for a parameter $s$, with high probability, any shortest path with at least $s$ vertices intersects $S$.

In this section, we first state the whole algorithm for completeness.
Then, we go on to show how to replace the randomization by a chain cover that deterministically guarantees to touch every long shortest path in the graph. Then, we can use our chain-based approximation to efficiently compute $\tdOp$ for every path with one endpoint on a chain, which gives a similar approximation guarantee to sampling (with an additional factor 3 due to \Cref{thm:chain_3_approx}).

\begin{algorithm}
    \textbf{Input:} Directed graph $G= (V,E)$ with $n$ vertices and parameter $s \in [n]$

    \textbf{Output:} $\tdOp$ estimate $\hat D$ such that $D - s \le \hat D \le D$, where $D$ is the diameter of $G$, whp
    
    \begin{enumerate}%[noitemsep, parsep=0pt, topsep=0pt]
        \item Sample $\OTilde(n / s)$  vertices $V'$ uniformly at random.
        \item Run a forwards BFS from the vertices in $V'$.
    \end{enumerate}%
    
        Return the height of the tallest BFS tree as $\hat D$.
    \caption{Randomized algorithm for computing $\tdOp$ in $\BigO(n/s \cdot m)$ time with $s$ additive error.}%
    \label{algo:random-sampling}%
\end{algorithm}

\thmRandomSampling*

\begin{proof}
  We consider \Cref{algo:random-sampling}.
  The algorithm takes time $\BigO(nm/k)$ as it requires $2n/k$ BFS runs.

    For the correctness proof, observe that as we always report the height of a BFS tree in $G$, we always have $\hat D \le D$.
    For the other inequality, first observe that if $D \le k$, the bound is trivial.
    Thus assume that there is a shortest path $P \colon a \rightsquigarrow b$ in $G$ with $|P| > k$. Then, with high probability, there is a vertex $c \in V'$ in the first $k$ vertices of $P$. 
    Therefore, one of the BFS trees from $c$ has height at least $D-k$.
\end{proof}

Now we show how to derandomize this simple sampling algorithm for DAGs. The idea is to substitute the sampling with a chain cover and then use the chain-based algorithm from \Cref{thm:chain_3_approx}.

An $\ell$-chain cover is similar to the minimum chain cover, but the parameter $\ell$ may be chosen arbitrarily to control the number of chains in the cover.
As a trade-off, we then might not be able to cover all vertices, but instead get a weaker guarantee that no long chains of uncovered vertices remain.
This will be enough for our use case here.

\begin{definition}[$\ell$-Chain Cover]
\label{def:chain-cover}
    An \emph{$\ell$-chain cover} is a set of at most $\ell$ vertex-disjoint chains $\mathcal C$ in $G$ such that for any path $P \subseteq G$, it holds that $\setsize{V(P) \setminus V(\mathcal C)} \le 2 n/ \ell$.
\end{definition}
Such a $\ell$-chain cover always exists and can be computed in near-linear time using MCMF algorithms~\cite{caceres_minimum_2023}.

\begin{lemma}[\cite{caceres_minimum_2023}]\label{lem:compute_l_cover}
    There is an algorithm that computes an $\ell$-chain cover of a directed graph $G$ in time $m^{1 + o(1)}$, for any $\ell$.
\end{lemma}

 The core observation is now that any path either touches a chain in the cover or has length at most $2n/\ell$. Thus, if we perform our algorithm that approximates the diameter of all paths touching a chain from the previous section, we will not lose too much compared to the randomized sampling algorithm.

\thmSamplingDerandomized*

\begin{proof}
    Our algorithm first computes a $2n / \ell$-chain cover. Then, for each chain $C$ in the chain cover, we run the algorithm from \Cref{thm:chain_3_approx} on $C$. Finally, we return the maximum answer $\hat{D}$ returned for any chain.

    The runtime follows immediately from \Cref{lem:compute_l_cover} and \Cref{thm:chain_3_approx}.
    
    For the correctness, let $a,b \in V$ be a pair of vertices with $d(a,b) = \td{G}$. If $d(a,b) \le \ell$, the bound holds trivially. Otherwise, a shortest path from $a$ to $b$ must intersect at least one chain from the chain cover. Let $C$ be a first such chain on a shortest path from $a$ to $b$ and let $v$ be the vertex where they intersect first.
    Then, by \Cref{thm:chain_3_approx} we get \[3\hat{D} \ge d(v,b) \ge d(a,b) - \ell.\] Furthermore, we only output valid distances, so the guarantees are as claimed.
\end{proof}

\subsection{Pivot-based Sampling}\label{sec:fineman}
Next, we present a near-linear-time approximation algorithm for $\tdOp$.
It provides a $n^{1/2+o(1)}$ multiplicative approximation factor.
Note that this is better than the $\Theta(k)$ to $\Theta(mn/k)$ trade-off shown by the (additive) approximation algorithms from the previous subsection. 
To achieve this, we can use a pivot-based approach pioneered by Jeremy Fineman in 2017 \cite{fineman_nearly_2017} and improved by Arun Jambulapati, Yang P. Jambulapati and Aaron Bernstein in 2019 \cite{DBLP:conf/focs/LiuJS19} designed for computing shortcut sets, modifying it to approximate $\tdOp$.
While their overall goal is to compute a shortcut set in parallel, they give a sequential routine to compute a shortcut set in $\Oish(m)$ time.
They are interested in computing a directed shortcut set with near-linear work and low depth.
We see that their shortcut set quality parameter will map to our approximation guarantee.
That is, the following theorem holds about \Cref{algo:ljs}.

\thmLjs*

The proof of this theorem can be seen as adaptation of Theorem 3 of \cite{DBLP:conf/focs/LiuJS19}.
Our stretch analysis maps to their hopbound analysis, but we need to whitebox it and show how each recursive call will affect the stretch.

\begin{algorithm}
    \textbf{Input:} Directed graph $G= (V,E)$, recursion level $r$

    \textbf{Output:} $\tdOp$ estimate $\hat D$, where $D$ is the diameter of $G$, whp
    
    \begin{enumerate}%[noitemsep, parsep=0pt, topsep=0pt]
        \item Base case: If the graph has a single vertex, return 0
        \item Sample vertices $S$ from $V$ with probability $p \in \Oish((\log n)^{r+1}/n)$ uniformly at random. Here $n$ is the original number of vertices at recursion level $0$.
      \item Run forwards and backwards BFS from the vertices in $S$.
      \item Each time, remove all vertices that are found in both the forward and backwards BFS from the same vertex in $S$.
      \item Partition the remaining vertices into sets $V_1, \dots$ such that all the vertices in each $V_i$ are reachable forwards and reachable backwards from the exact same vertices from $S$.
        \item Recursively call this algorithm on the graphs $G[V_1], \dots$, increasing $r$ by one.
    \end{enumerate}%
    
    Report the largest distance seen at any recursion level as the diameter estimate $\hat D$.
    \caption{Algorithm of \cite{DBLP:conf/focs/LiuJS19} adapted for $\tdOp$}%
    \label{algo:ljs}%
\end{algorithm}

For simplicity of presentation, we have left out how the shortcuts are added as it does not improve recursive behavior or the sequential running time.

\paragraph{Running Time.} By comparing to the original algorithm, it is easy to verify that we do not perform extra work (except tracking the deepest BFS performed, which is straightforward in the presented recursion).

\paragraph{Correctness.}
Verifying that this algorithm correctly approximates the diameter requires a closer look.
As part of their shortcut set quality analysis, \cite{DBLP:conf/focs/LiuJS19} establish the following lemma, which we will rely on to establish out approximation guarantee.

\begin{lemma}[Inductive application of Lemma 4.4 in \cite{DBLP:conf/focs/LiuJS19}]
  \label{lem:ljs-core-lemma}
  Let $P$ be a path in $G$. Then, whp, we can partition $P$ into $n^{1/2+o(1)}$ subpaths $P_1, \dots, P_\ell$ such that for each $P_i$ there is a recursive application of \Cref{algo:ljs} where there is a vertex $u$ in its set $S$ such that $(P_i)_1$ can reach $u$ and $u$ can reach $(P_i)_{\setsize{P_i}}$.
\end{lemma}

The following observation follows from the triangle inequality.

\begin{observation}
  \label{lem:pivot-2-approx}
  Let $P=v_1, \dots, v_{\setsize{P}}$ be a shortest path and let $u$ be a vertex such that $v_1$ can reach $u$ and $u$ can reach $v_{\setsize{P}}$. Then $\max(d(v_1, u), d(u, v_{\setsize{P}})) \ge d(v_1, v_{\setsize{P}})/2$.
\end{observation}

To show that the algorithm is useful for diameter estimation, we need to examine how the recursive splitting will affect the distances between any two vertices.
Towards this, the following lemma is key.

\begin{lemma}
  Let $u$ and $v$ such that $u$ can reach $v$.
  If $u$ and $v$ get assigned to the same sub-instance by \Cref{algo:ljs}, then so do all vertices between $u$ and $v$ (i.e., all vertices on any $u$ to $v$ path).
\end{lemma}

\begin{proof}
  Assume that is not the case for some vertex $w$ in between $u$ and $v$.
  Then there is a pivot vertex $s \in S$ such that the reachability between $s$ and $w$ differs from the reachability between $s$ and $u$ and $v$ (which is the same).

  Assume that $w$ can be reached from $s$ but $u$ and $v$ cannot.
  Since $w$ is in between $u$ and $v$, $w$ can reach $v$ and thus $v$ can be reached from $s$, which contradicts our assumption.
  The other three cases are (1) $w$ can reach $s$, but $u$ and $v$ cannot, (2) $u$ and $v$ can be reached from $s$ but $w$ cannot, and (3) $u$ and $v$ can reach $s$, but $w$ cannot.
  These can be argued analogously.
\end{proof}

This allows us to argue how the distances between vertices change when we create the sub-instances in the algorithm.
In particular, we have the following.

\begin{corollary}
  Let $G'$ be an arbitrary sub-instance created in \Cref{algo:ljs} and $G$ be the original graph.
  Then, for all $u,v \in V(G')$, we have $d_{G'}(u,v) = d_G(u,v)$.
  \label{cor:ljs-distance-preservation}
\end{corollary}

That is, if two vertices get assigned to the same sub-instance, their distance is preserved.
Thus, if we think about two different levels of recursion, some distances in the deeper level are $\infty$ because the vertices were disconnected, but all remaining non-infinite distances are the same as in the original graph.
With this, we have all pieces in hand to prove the theorem of this section.

\begin{proof}[Proof of \Cref{thm:ljs}]
  Applying \Cref{lem:ljs-core-lemma} a diameter path $P$, we learn that whp $P$ can be split into at most $n^{1/2+o(1)}$ subpaths such that each subpath in fully contained in sub-instances of $G'$ such that one of its pivots can reach the start and the end of the respective subpath.
  By the pigeon hole principle, one of these subpaths must have at least a $1/n^{1/2+o(1)}$ fraction of the length of the path.
  By \Cref{cor:ljs-distance-preservation}, the length of this subpath is preserved in this sub-instance.
  Applying \Cref{lem:pivot-2-approx} to the pivot in the respective sub-instance that can be reached by the start and reach the end of this subpath, yields the desired approximation guarantee as in the algorithm we will run BFS from it.
\end{proof}

\fi

%%
%% Bibliography
%%

%% Please use bibtex, 

\bibliography{papers.bib}

% \appendix

\end{document}